\documentclass[submission,creativecommons]{eptcs}
\providecommand{\event}{LSFA 26} 

\usepackage{iftex}

\ifpdf
  \usepackage{underscore}         
  \usepackage[T1]{fontenc}        
\else
  \usepackage{breakurl}           
\fi

\usepackage{proof}
\usepackage{amsmath}
\usepackage{amsfonts}
\usepackage{graphicx}
\usepackage{enumitem}
\usepackage{suffix}

\newtheorem{definition}{Definition}
\newtheorem{lemma}{Lemma}
\newtheorem{proposition}{Proposition}

\newcommand{\mybox}{\raisebox{3.5pt}{\framebox[7pt]{\rule{0mm}{0.3pt}}}}

\def\qed{\ifmmode\mybox\else{\unskip\nobreak\hfil
\penalty50\hskip1em\null\nobreak\hfil\mybox
\parfillskip=0pt\finalhyphendemerits=0\endgraf}\fi}

\newenvironment{proof}{%

\vspace{.7ex}
\noindent{\it Proof.} \rm}{\qed}

\def\joinrelm{\mathrel{\mkern-3mu}}
\def\relbar{\mathrel{\smash-}}
\def\tailpiece{\mathchoice
  {\vrule height 1ex width 0.3ex depth -.1ex}%
  {\vrule height 1ex width 0.3ex depth -.1ex}%
  {\vrule height 0.7ex width 0.3ex depth -.1ex}%
  {\vrule height 0.7ex width 0.3ex depth -.1ex}}
\def\seqsym{\mathrel{\tailpiece\joinrelm\relbar}}

\newcommand{\bt}{{\scriptstyle\bot}}
\newcommand{\iabs}[1]{(\lambda\,#1)}
\WithSuffix\newcommand\iabs*[1]{\lambda\,#1}
\newcommand{\iapp}[2]{(#1\,#2)}
\WithSuffix\newcommand\iapp*[2]{#1\,#2}

\newcommand{\sbstt}[3]{#1[#2 := #3]}

\newcommand{\liftt}[3]{{\uparrow}^{#2}_{#3}(#1)}

\newcommand{\tseq}[3]{#1 \, \seqsym \, #2 : #3}
\newcommand{\tytm}[4]{\{#1\}\;#2:#3\;\{#4\}}

\newcommand{\der}[1]{\vdash_{#1}}
\newcommand{\Sst}{{\textsc{st}}}

\newcommand{\Sfll}{{\textsc{fll}}}
\newcommand{\Sll}{{\textsc{ll}}}

\newcommand{\ltm}{\mathord{\mathbf{lt}}}
\newcommand{\cst}{\mathord{\mathbf{c}}}
\newcommand{\ty}{\mathord{\mathbf{ty}}}
\newcommand{\qty}{\mathord{\mathbf{qty}}}
\newcommand{\aty}{\mathord{\mathbf{a}}}
\newcommand{\env}{\mathord{\mathbf{env}}}
\newcommand{\fenv}{\mathord{\mathbf{fenv}}}

\newcommand{\occ}[2]{\mathord{\mathrm{occ}}(#1,#2)}
\newcommand{\fv}[1]{\mathord{\mathrm{fv}}(#1)}
\newcommand{\minenv}{\mathord{\bot\hspace{-1.55ex}\bot}}

\title{A Typing System for the Linear Lambda-Calculus\\ in \mbox{de Bruijn} Notation%
\thanks{%
The research leading to these results has received funding from the European Research Council (ERC)
under the European Union's Ninth Framework Programme Horizon Europe (ERC Synergy Project Malinca,
Grant Agreement n. 101167526).}} 

\author{%
Philippe de Groote
\institute{LORIA, UMR 7503,\\ Universit{\'e} de Lorraine, CNRS, Inria,\\ 54000 Nancy, France}
\and
Vincent Tourneur
\institute{LORIA, UMR 7503,\\ Universit{\'e} de Lorraine, CNRS, Inria,\\ 54000 Nancy, France}
}

\newcommand{\titlerunning}{A Typing System for the Linear Lambda-Calculus in \mbox{de Bruijn} Notation}
\newcommand{\authorrunning}{Ph. de Groote and V. Tourneur}

\hypersetup{
  bookmarksnumbered,
  pdftitle    = {\titlerunning},
  pdfauthor   = {\authorrunning},
  pdfsubject  = {LFSA26 submission}, 
  pdfkeywords = {Linear lambda-calculus, Type sytem}
}

\begin{document}

\maketitle

\begin{abstract}
We introduce a typing system that is particularly well suited for typing the linear $\lambda$-calculus
in de Bruijn notation. This typing discipline, which is reminiscent of Hodas' and Miller's model of resource
consumption, guarantees that any well-typed term is linear without the need for an occurrence check.
We then establish the subject reduction property.
\end{abstract}

\section{Introduction}

Typed lambda calculus plays an important role in several domains of computer science.
It is the basis of typed functional programming, serves as a foundation to proof assistants and logical
frameworks, and also plays its part in type theoretic approaches to computational linguistics.

One of the problems one faces when implementing the lambda calculus is managing the names of the bound variables
and their possible renaming. A common way to get around this is to rely on nameless encoding of the bound
variables.
De Bruijn notation~\cite{deBruijn:indices} is such an encoding.
It is a canonical way of expressing $\lambda$-terms without using names
for representing the bound occurrences of variables and therefore avoiding all the difficulties related
to alpha-conversion, such as possible clashes between free and bound occurrences of a same variable.
This notation has been designed for the automatic manipulation of $\lambda$-terms and is particularly
well-suited for implementation.

In de Bruijn notation, the bound occurrences of a variable are represented by natural numbers
(called \emph{de Bruijn indices}) that may be interpreted as pointers to their binding lambda.
As for the free occurrences of a variable, they may be interpreted as pointers to a position in some kind of
external environment of variables (typically, a typing environment in the case of a
typed $\lambda$-calculus).
For this reason, the use of de Bruijn indices works quite smoothly with application typing rules of the
following kind:
\begin{align*}
\infer{
\tseq{\Gamma}{\iapp{t}{u}}{\beta}}{
\tseq{\Gamma}{t}{\alpha \rightarrow \beta}
&
\tseq{\Gamma}{u}{\alpha}}
\end{align*}
where the free occurrences of a same variable in both $t$ and $u$ are interpreted as a pointer to some position
in the typing environment $\Gamma$.
Consequently, the same free variable can occur in both t and u, and the above typing rule implicitly includes
possible contraction rules. 

Using the vocabulary of linear logic~\cite{Girard:LL}, the above typing rule is called an \emph{additive}
rule because the environment that appears in the conclusion (namely, $\Gamma$) also appears in both premises.
This contrasts drastically with the application typing rule that is used in the case of the linear
$\lambda$-calculus, for which contraction is forbidden with the consequence that each variable occurring
in a term (free or bound) must occur exactly once:
\begin{align*}
\infer{
  \tseq{\Gamma,\Delta}{\iapp{t}{u}}{\beta}}{
  \tseq{\Gamma}{t}{\alpha \rightarrow \beta}
&
  \tseq{\Delta}{u}{\alpha}}
\end{align*}
with the proviso that the domains of $\Gamma$ and $\Delta$ must be disjoint.

In the words of linear logic, the linear application typing rule is called a \textit{multiplicative rule}.
Such multiplicative rules do not fit smoothly the use of de Bruijn indices, because the variable corresponding
to a given position in the environment $\Gamma$ differs from the variable corresponding to this same position
in the environment $\Delta$. For this reason, de Bruijn notation is not well suited for multiplicative rules
and therefore does not seem to be suitable for the case of the linear lambda calculus.

The purpose of this paper is to circumvent the above difficulty and to provide a typing system for the linear
$\lambda$-calculus that accommodates $\lambda$-terms in de Bruijn notation nicely. To this end, we introduce
a notion of \emph{fragmentary environment} that allows multiplicative rules to be simulated by additive rules.

The rest of the paper is organized as follows:
\begin{itemize}
\item In the next section, we briefly explain how the de Bruijn indices work.
  We then expose some of the mathematical preliminaries that are necessary for understanding of the paper,
  including the simply typed $\lambda$-calculus in de Bruijn notation and the notion of a linear
  $\lambda$-term.
\item In Section 3, we introduce the notion of a fragmentary environment. We then define some algebraic
  operations on fragmentary environments, operations that will turn out to be useful when establishing the
  mathematical properties of our main typing system. We also define a typing system for the linear $\lambda$-calculus, based
  on fragmentary environments. 
\item In Section 4, we outline the main typing system of the paper. This system allows the linear $\lambda$-terms
  (and only the linear $\lambda$-terms) in de Bruijn notation to be typed. It does not rely on any explicit
  proviso (such as an occurrence unicity condition, or a domain disjointness condition), and is reminiscent
  of Hodas' and Miller's model of resource consumption~\cite{Hodas-Miller} (from which we adopt the curly bracket notation).
\item In Section 5, we review how $\beta$-reduction is implemented using de Bruijn indices.
\item In Section~6, we prove that the typing system in Section 4 satisfies the subject reduction property.
\item Finally, in Section 7, we discuss related work and conclude.
\end{itemize}

\section{The simply typed lambda-Calculus in de Bruijn Notation}

In de Bruijn notation, occurrences of $\lambda$-variables are encoded by means of natural numbers called
de Bruijn indices.
In the case of a bound occurrence of a variable, the de Bruijn index corresponds to the number of $\lambda$'s
lying on the path that links, in the parse tree of the term, the bound occurrence to its binding $\lambda$.
For instance, the following $\lambda$-term (in usual notation):
\begin{align}\label{ltusual}
\lambda x y. \, y\,x\,(\lambda z.\,z\,x)
\end{align}
is encoded as follows:
\begin{align}\label{ltdb}
\iabs{\iabs{\iapp{\iapp{0}{1}}{\iabs{\iapp{0}{2}}}}}.
\end{align}
This is illustrated by Figure \ref{ltpt} that displays the parse tree of $\lambda$-term~(\ref{ltdb}).

\begin{figure}[ht]
\fbox{%
\begin{minipage}{.975\textwidth}
\mbox{}
\vspace{\parsep}
\begin{center} 
\includegraphics[width=3cm]{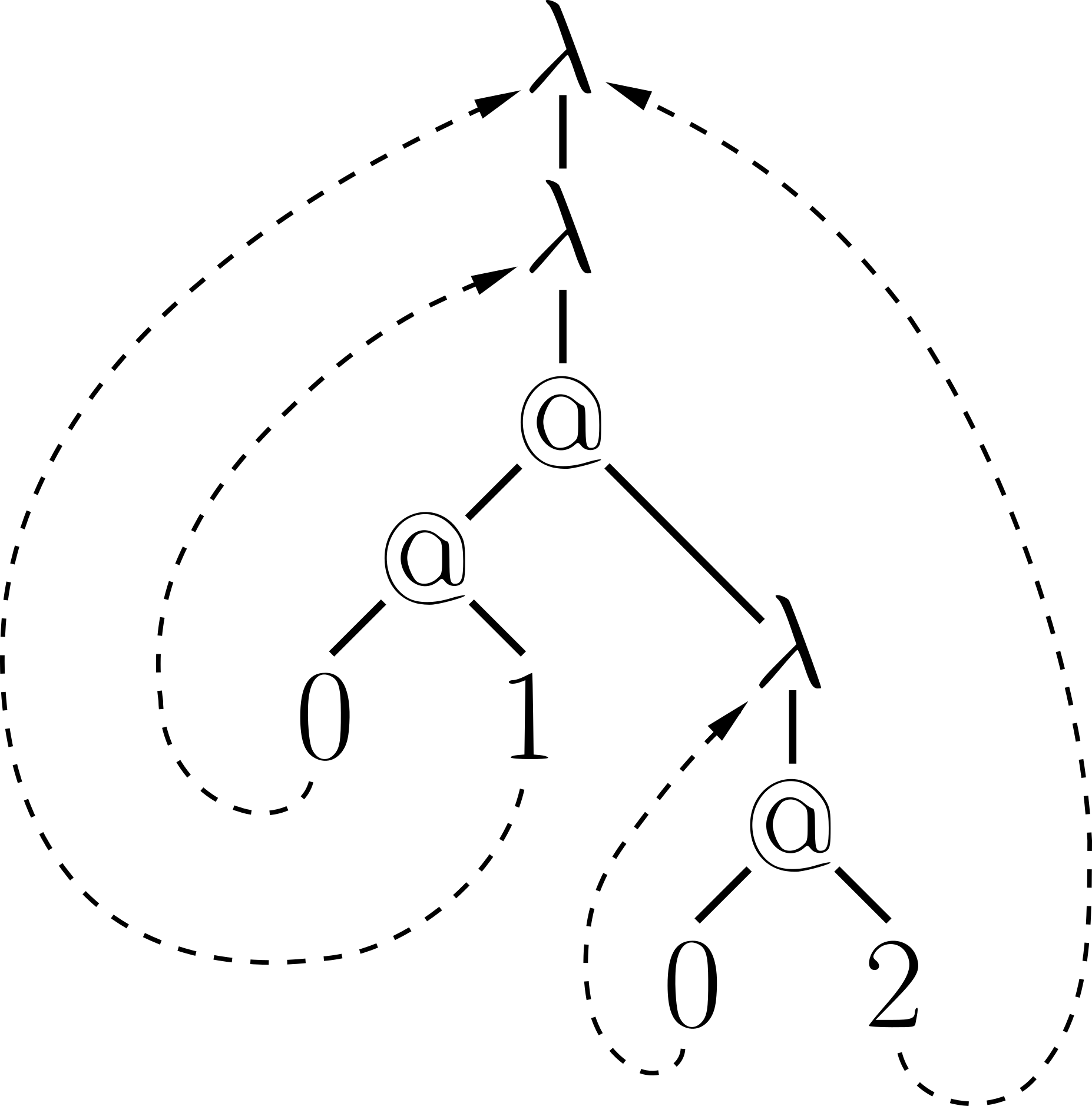}
\end{center}
\vspace{\parsep}
\end{minipage}
}
\caption{A $\lambda$-term in de Bruijn notation}\label{ltpt}
\end{figure}

When using de Bruijn notation, a $\lambda$-term consists of either a constant, an index, an abstraction,
or an application.
This is spelled out by the following definition.

\begin{definition}\label{ltdef}
Let $\cst$ be a set of constants. The set of $\lambda$-terms, $\ltm$, is inductively defined by the following
grammar:
\[
  \ltm  \; ::= \; \cst \;\;|\;\;  \mathrm{i}  \;\;|\;\; \iapp{\ltm}{\ltm} \;\;|\;\;  \iabs{\ltm}
  \quad\quad(\textrm{where }\mathrm{i} \in \mathbb{N})
\]
\end{definition}

From now on, when speaking of a $\lambda$-term, we will mean a $\lambda$-term in de Bruijn notation, i.e.,
a $\lambda$-term obeying Definition~\ref{ltdef}.
In a similar vein, when speaking of a variable or of the occurrence of a variable, we will mean the
corresponding index.

Let us now specify the usual typing system of the simply typed $\lambda$-calculus~\cite{Church:STT} for
$\lambda$-terms in de Bruijn notation.
We first remind the reader of the definition of a simple type.
Then we introduce the notion of a typing environment.

\begin{definition}
Let $\aty$ be a set of atomic types. The set of simple types, $\ty$, is inductively defined by
the following grammar:
\[
  \ty \; ::= \; \aty \;\; | \;\; (\ty \rightarrow \ty)
\]
\end{definition}

\begin{definition}
A typing environment (or an environment, for short) is defined to be a finite sequence of simple types.
Accordingly, the set of environments, $\env$,
is defined as follows:
\begin{align*}
\env = \bigcup_{n\in\mathbb{N}} \ty^n
\end{align*}
\end{definition}
We let lowercase Greek letters (from the beginning of the alphabet) range over simple types,
and uppercase Greek letters range over environments. Given an environment $\Gamma$, we write
$|\Gamma|$ for its length. For $0\leq i < |\Gamma|$, we write $\Gamma(i)$ for the $i^\textit{th}$ component
of $\Gamma$, starting to count from $0$ and numbering the components from right to left.
For instance, if $\Gamma = (\alpha,\beta,\gamma)$, $\Gamma(0)=\gamma$, $\Gamma(1)=\beta$, and $\Gamma(2)=\alpha$.

$\lambda$-terms are assigned simple types by means of typing judgements of the form
$\tseq{\Gamma}{t}{\alpha}$ where $\Gamma\in\env$, $t\in\ltm$, and $\alpha\in\ty$.

Let $\tau$ be a function that assigns to each constant $c\in\cst$ a type $\tau(c)\in\ty$.
The typing judgements are derived according to the axioms
and rules given in Figure \ref{tsstlc}.

\begin{figure}[ht]
\fbox{%
\begin{minipage}{.975\textwidth}
\mbox{}
\vspace{\parsep}
\begin{center}
\raisebox{1.5ex}{$%
\tseq{\Gamma}{c}{\tau(c)}\;\;\mbox{(\textsc{const})}
\quad\quad
\tseq{\Gamma, \alpha}{0}{\alpha}\;\;\mbox{(\textsc{var})}
$}
\quad\quad
$
\infer[\mbox{(\textsc{weak})}]{
\tseq{\Gamma, \beta}{\mathrm{i}+1}{\alpha }}{
\tseq{\Gamma}{\mathrm{i}}{\alpha }}
$
\end{center}
\begin{center}
$
\infer[\mbox{(\textsc{app})}]{
\tseq{\Gamma}{\iapp{t}{u}}{\beta}}{
\tseq{\Gamma}{t}{\alpha \rightarrow \beta}
&
\tseq{\Gamma}{u}{\alpha}}
\quad\quad
\infer[\mbox{(\textsc{abs})}]{
\tseq{\Gamma}{\iabs{t}}{\alpha \rightarrow \beta}}{
\tseq{\Gamma, \alpha}{t}{\beta}}
$
\end{center}
\vspace{\parsep}
\end{minipage}
}
\caption{Typing system for the simply-typed $\lambda$-calculus}\label{tsstlc}
\end{figure}

When a typing judgement, $\tseq{\Gamma}{t}{\alpha}$, is derivable according to the system of Figure~\ref{tsstlc},
we write $\der{\Sst} \tseq{\Gamma}{t}{\alpha}$.

Note that in Rule~(\textsc{var}), the type assigned to index 0 is the rightmost type occurring in the
environment.
Note also that Rule~(\textsc{weak}) has the effect of incrementing the index.
This exemplifies the fact that an index that stands for a free variable corresponds to a position in the
environment (starting from zero and numbering the positions from right to left).

We now turn to the question of what it means to be linear for a lambda term in de Bruijn notation.
Using the usual notation, with named variables, a $\lambda$-term $t$ is linear if and only if:
\begin{itemize}
\item every subterm of $t$ of the form $\lambda x.\,u$ is such that $x$ has exactly one free occurrence in $u$;
\item every variable has at most one free occurrence in $t$.
\end{itemize}
This corresponds to the $\lambda$-terms that can be typed in the implicative
fragment of linear logic~\cite{Benton:TLCA93}.

In order to adapt the above definition to the case of $\lambda$-terms in de Bruijn notation, we must first
define the notion of a free occurrence of a de Bruijn index in a $\lambda$-term.
More precisely, what we define is the number of free occurrences of a given index $i$ in a given
$\lambda$-term $t$.

\begin{definition}
Let $i\in\mathbb{N}$ and $t\in\ltm$. The number of free occurrences of index $i$ in $t$, in notation $\occ{i}{t}$, is
inductively defined as follows:
\begin{enumerate}[label={\rm \roman*.}, labelsep*=2ex, itemsep=0pt]
\item $\occ{i}{c} = 0$,\, for $c\in\cst$
\item $\occ{i}{i} = 1$
\item $\occ{i}{j} = 0$,\, for $j\in\mathbb{N}$ and $i\not=j$
\item $\occ{i}{\iapp{t_1}{t_2}} = \occ{i}{t_1} + \occ{i}{t_2}$
\item $\occ{i}{\iabs{t_1}} = \occ{i+1}{t_1}$
\end{enumerate}
\end{definition}

We then define the set of free variables of a term to be the set of indices whose numbers of occurrences are
nonzero.
\begin{definition}
Let $t\in\ltm$. The set of free variables of $t$, in notation $\fv{t}$, is defined as follows:
\begin{align*}
\fv{t} = \{i\in\mathbb{N}\;|\;\occ{i}{t} \not= 0\}
\end{align*}
\end{definition}

We are now in a position of defining what linearity signifies in the case of a $\lambda$-term in de Bruijn
notation.

\begin{definition}\label{lltdef}
Let $t\in\ltm$. $t$ is said to be a \emph{linear $\lambda$-term} if and only if the three following conditions hold:
\begin{enumerate}[label={\rm \roman*.}, ref={\rm \roman*}, labelsep*=2ex, itemsep=0pt]
\item every subterm of $t$ of the form $\iabs{u}$ is such that $\occ{0}{u} = 1$;\label{C1}
\item for every $i\in\fv{t}$, $\occ{i}{t} = 1$;\label{C2}
\item for every $i, j \in \mathbb{N}$, if $i\in\fv{t}$ and $j<i$, then $j \in \fv{t}$.\label{C3}
\end{enumerate}
\end{definition}

The intuition behind the third condition in the above definition can be paraphrased as follows:
\textit{every variable declared in the environment must occur in the lambda term.}
This intuition may seem a little bit subtle because the notion of a linear $\lambda$-term would not only depend
on the term itself but also on a possible declaration environment.

To illustrate it, consider the following two typing judgements (expressed using the usual notation with named
variables):
\begin{align}
&\tseq{y:\alpha\rightarrow\beta}{\lambda x.\,y\,x}{\alpha\rightarrow\beta}\label{tj1}\\
&\tseq{y:\alpha\rightarrow\beta, z:\gamma}{\lambda x.\,y\,x}{\alpha\rightarrow\beta}\label{tj2}
\end{align}
The first typing judgement, (\ref{tj1}), is a valid judgement of the typed linear $\lambda$-calculus.
Therefore, according to this typing judgement, $\lambda x.\,y\,x$ is considered to be a linear $\lambda$-term.
By contrast, typing judgement~(\ref{tj2}) is not a valid judgement of the typed linear $\lambda$-calculus because
the variable $z$, which is declared in the environment, does not occur in the typed $\lambda$-term.
Consequently, according to this second typing judgement, $\lambda x.\,y\,x$ is not considered to be a linear
$\lambda$-term.
Now, if one expresses judgements (\ref{tj1}) and (\ref{tj2}) using de Bruijn notation, the two involved
$\lambda$-terms are no longer syntactically equal:
\begin{align}
&\tseq{\alpha\rightarrow\beta}{\iabs{\iapp{1}{0}}}{\alpha\rightarrow\beta}\label{tjdb1}\\
&\tseq{\alpha\rightarrow\beta, \gamma}{\iabs{\iapp{2}{0}}}{\alpha\rightarrow\beta}\label{tjdb2}
\end{align}
We then invite the reader to check that $\iabs{\iapp{1}{0}}$ obeys Definition~\ref{lltdef} and is therefore a
linear $\lambda$-term, while $\iabs{\iapp{2}{0}}$ is not because it violates condition \ref{C3}
of Definition~\ref{lltdef}.

In the sequel, $\lambda$-terms that only obey conditions \ref{C1} and \ref{C2} of Definition~\ref{lltdef} will
be called \emph{quasi-linear $\lambda$-terms}.

\section{Fragmentary typing environments}

The distinction between linear and quasi-linear $\lambda$-terms highlights one of the technical problems one
faces when trying to work with linear $\lambda$-terms in de Bruijn notation: in general, the subterms of a
linear $\lambda$-term are not themselves linear; they are only quasi-linear.
As an illustration, consider the following typing judgement (in usual notation):
\begin{align}\label{ex2}
\tseq{}{\lambda x y z.\,y\,(x\,z)}{(a\rightarrow b)\rightarrow((b\rightarrow c)\rightarrow(a\rightarrow c))}
\end{align}
Using de Bruijn notation, typing judgement (\ref{ex2}) is as follows:
\begin{align}\label{ex2dbn}
  \tseq{}
    {\iabs{\iabs{\iabs{\iapp{1}{\iapp{2}{0}}}}}}
    {(a\rightarrow b)\rightarrow((b\rightarrow c)\rightarrow(a\rightarrow c))}
\end{align}
The body of term~(\ref{ex2dbn}), namely $\iapp{1}{\iapp{2}{0}}$, is made of two subterms, $1$ and $\iapp{2}{0}$,
that are both quasi-linear but not linear.
Therefore, a syntax-oriented typing system for linear $\lambda$-terms must accommodate quasi-linear
$\lambda$-terms.
Now, let us focus on the rule that allows the body of term~(\ref{ex2dbn}) to be typed:
\begin{align}\label{exapp}
  \infer[\mbox{(\textsc{app})}]{
    \tseq{a\rightarrow b, b\rightarrow c, a}{\iapp{1}{\iapp{2}{0}}}{c}
  }{
    \tseq{a\rightarrow b, b\rightarrow c, a}{1}{b\rightarrow c}
    &
    \tseq{a\rightarrow b, b\rightarrow c, a}{\iapp{2}{0}}{b}
  }
\end{align}
In order to derive the first premise of rule~(\ref{exapp}), the only needed typing information is the type that
appears at position 1 in the typing environment. Similarly, in order to derive the second premise, one only
needs the types appearing at positions 0 and 2. The idea of a fragmentary environment is the one of a typing
environment from which the irrelevant typing information has been discarded. Applying this idea to our current
example, using $\bt$ as a placeholder for the discarded irrelevant typing information, we obtain the following
typing rule:
\begin{align}\label{exappf}
  \infer[\mbox{(\textsc{app})}]{
    \tseq{a\rightarrow b, b\rightarrow c, a}{\iapp{1}{\iapp{2}{0}}}{c}
  }{
    \tseq{\bt, b\rightarrow c, \bt}{1}{b\rightarrow c}
    &
    \tseq{a\rightarrow b, \bt, a}{\iapp{2}{0}}{b}
  }
\end{align}

In order to put the idea illustrated by rule~(\ref{exappf}) to work, we need to formalize the notion of a
fragmentary environment and to define the operation of merging two fragmentary environments (an operation that
we will define as an addition).

\begin{definition}
Let $\bt$ be a fresh symbol (i.e., a symbol such that $\bt\not\in\aty$).
The set of quasi-types, $\qty$, is defined to be $\ty \cup \{\bt\}$. 
\end{definition}

We let lowercase Greek letters from the end of the alphabet ($\varphi, \psi, \omega$) range over quasi-types.

\begin{definition}
A fragmentary environment is defined to be a finite sequence of quasi-types.
Accordingly, the set of fragmentary environments, $\fenv$, is defined as follows:
\begin{align*}
\fenv = \bigcup_{n\in\mathbb{N}} \qty^n
\end{align*}
\end{definition}

We adopt the same notational conventions for fragmentary environments as for typing environments.

We provide the set of quasi-types with two partial operations: addition and subtraction.
These operations will then be lifted at the level of the fragmentary environments and will be used for their
merging.

\begin{definition}
Let $\omega \in \qty$. The addition of quasi-types is the smallest partial operation that satisfies the
following two equations:
\begin{align*}
  \bt + \omega = \omega \quad\quad
  \omega + \bt = \omega
\end{align*}
\end{definition}

\begin{definition}
Let $\omega \in \qty$. The subtraction of quasi-types is the smallest partial operation that satisfies the
following two equations:
\begin{align*}
  \omega - \bt = \omega \quad\quad
  \omega - \omega = \bt
\end{align*}
\end{definition}

We also provide the set of quasi-types with a flat partial order.

\begin{definition}
Let $\psi, \omega \in \qty$. $\psi\sqsubseteq\omega$ if and only if $\psi=\bot$ or $\psi=\omega$.
\end{definition}

We now lift the addition and subtraction operation, as well as the order relation, to the level of fragmentary
environments. This is simply done componentwisely.

\begin{definition}
Let $\Gamma, \Delta \in \fenv$.
The addition of fragmentary environments is the smallest partial operation such that if $\Gamma + \Delta$ is defined then
\begin{enumerate}[label={\rm \roman*.}, labelsep*=2ex, itemsep=0pt]
  \item $|\Gamma + \Delta| = |\Gamma| = |\Delta|$;
  \item for all $0\leq i < |\Gamma|$, $(\Gamma + \Delta)(i) = \Gamma(i) + \Delta(i)$.
\end{enumerate}
\end{definition}

\begin{definition}
Let $\Gamma, \Delta \in \fenv$.
The subtraction of fragmentary environments is the smallest partial operation such that if $\Gamma - \Delta$ is defined then
\begin{enumerate}[label={\rm \roman*.}, labelsep*=2ex, itemsep=0pt]
  \item $|\Gamma - \Delta| = |\Gamma| = |\Delta|$;
  \item for all $0\leq i < |\Gamma|$, $(\Gamma - \Delta)(i) = \Gamma(i) - \Delta(i)$.
\end{enumerate}
\end{definition}

\begin{definition}
Let $\Gamma, \Delta \in \fenv$. $\Gamma\sqsubseteq\Delta$ if and only if
$|\Gamma| = |\Delta|$ and for all $0\leq i < |\Gamma|$, $\Gamma(i) \sqsubseteq \Delta(i)$.
\end{definition}

We say that two fragmentary environments, $\Gamma$ and $\Delta$, are \emph{disjoint}
(in notation, $\Gamma\bowtie\Delta$) if and only if $\Gamma+\Delta$ is defined.
The minimal elements of $\fenv$ are fragmentary environments whose every component is $\bt$, i.e.,
environments of the form $\bt,\bt,\ldots,\bt$. We write $\mathbf{\minenv}$ for these minimal environments.

We now adapt the typing system of Figure~\ref{tsstlc} in order to allow only quasi-linear terms to be typed.
To this aim, we exploit the idea we presented at the beginning of this section, i.e., discarding the irrelevant
typing information from the environments. The resulting system is given in Figure~\ref{tsfllc}.

\begin{figure}[ht]
\fbox{%
\begin{minipage}{.975\textwidth}
\mbox{}
\vspace{\parsep}
\begin{center}
\raisebox{1.5ex}{$%
\tseq{\minenv}{c}{\tau(c)}\;\;\mbox{(\textsc{const})}
\quad\quad
\tseq{\minenv, \alpha}{0}{\alpha}\;\;\mbox{(\textsc{var})}
$}
\quad\quad
$
\infer[\mbox{(\textsc{weak})}]{
\tseq{\Gamma, \bt}{\mathrm{i}+1}{\alpha }}{
\tseq{\Gamma}{\mathrm{i}}{\alpha }}
$
\end{center}
\begin{center}
$
\infer[\mbox{(\textsc{app})}]{
\tseq{\Gamma+\Delta}{\iapp{t}{u}}{\beta}}{
\tseq{\Gamma}{t}{\alpha \rightarrow \beta}
&
\tseq{\Delta}{u}{\alpha}}
\quad\quad
\infer[\mbox{(\textsc{abs})}]{
\tseq{\Gamma}{\iabs{t}}{\alpha \rightarrow \beta}}{
\tseq{\Gamma, \alpha}{t}{\beta}}
$
\end{center}

\quad
\textsf{\textbf{Proviso}}: in Rule (\textsc{app}), $\Gamma \bowtie \Delta$.

\vspace{\parsep}
\mbox{}
\end{minipage}
}
\caption{Typing system with fragmentary environments for the linear $\lambda$-calculus}
\label{tsfllc}
\end{figure}

When a typing judgement, $\tseq{\Gamma}{t}{\alpha}$, is derivable according to the typing system
of Figure~\ref{tsfllc}, we write $\der{\Sfll} \tseq{\Gamma}{t}{\alpha}$.
The next lemma, which can be established by a straightforward induction, sheds some light on the rationale
for this typing system.

\begin{lemma}
Let $\Gamma\in\fenv$, $t\in\ltm$, and $\alpha\in\ty$ be such that $\der{\Sfll} \tseq{\Gamma}{t}{\alpha}$.
For every $i$ such that $0\leq i <|\Gamma|$, if $\Gamma(i)\in\ty$ then $\occ{i}{t} = 1$, otherwise
$\Gamma(i)=\bot$ and $\occ{i}{t} = 0$.
\end{lemma}

Using this lemma, one can establish the following two propositions, which state that all and only the
quasi-linear simply-typed lambda-terms can be typed according to the typing system of Figure~\ref{tsfllc}.

\begin{proposition}\label{fst2fll}
Let $\Gamma\in\env$, $t\in\ltm$, and $\alpha\in\ty$ be such that $t$ is quasi-linear and
$\der{\Sst} \tseq{\Gamma}{t}{\alpha}$.
Then, there exists $\Delta\in\fenv$ such that $\Delta\sqsubseteq\Gamma$ and
$\der{\Sfll} \tseq{\Delta}{t}{\alpha}$.
\end{proposition}

\begin{proposition}\label{fll2fst}
Let $\Gamma\in\fenv$, $t\in\ltm$, and $\alpha\in\ty$ be such that $\der{\Sfll} \tseq{\Gamma}{t}{\alpha}$.
Then, $t$ is quasi-linear and there exists $\Delta\in\env$ such that $\Gamma\sqsubseteq\Delta$ and
$\der{\Sst} \tseq{\Delta}{t}{\alpha}$.
\end{proposition}

We conclude this section by stating a few technical properties that will prove useful later on.
These properties, which concern the quasi types, the fragmentary environments, and the operations of addition
and subtraction, take the form of equalities.
We first state them for the quasi-types.

\begin{lemma}\label{lemma:qtSimpleEq}
Let $\varphi, \psi, \omega \in \qty$. Then, the following equalities hold, if all operations are defined:
\begin{enumerate}[label={\rm (\alph*)}, ref={\rm \alph*}, labelsep*=2ex, itemsep=0pt]
  \item $\varphi + (\psi - \omega) = \psi - (\omega - \varphi)$ \label{qtSimpleEq:eq1}
  \item $\varphi + (\psi - \varphi) = \psi$ \label{qtSimpleEq:eq2}
  \item $\varphi - (\varphi - \psi) = \psi$ \label{qtSimpleEq:eq3}
\end{enumerate}
\begin{proof}
We prove (\ref{qtSimpleEq:eq1}), the proofs of the other equalities being similar.
Either $\varphi = \bot$ or $\varphi\in\ty$. In the first case, both sides of the equation simplify to
$\psi - \omega$. In the second case, since $\omega - \varphi$ is defined, we have that $\omega = \varphi$.
Then, since $\psi - \omega$ is defined, we have that $\psi = \omega$. Hence, the left-hand side of the equation
simplifies to $\varphi$, and the right-hand side to $\psi$, which establishes the equality because
$\varphi=\omega$ and $\omega=\psi$.
\end{proof}
\end{lemma}

Since the addition and subtraction of fragmentary environments are defined componentwise, we obtain as a
corollary that the same equalities hold for fragmentary environments.

\begin{lemma}\label{lemma:envSimpleEq}
Let $\Gamma, \Delta, \Theta \in \fenv$. Then, the following equalities hold, if all operations are defined:
\begin{enumerate}[label={\rm (\alph*)}, labelsep*=2ex, itemsep=0pt]
  \item $\Gamma + (\Delta - \Theta) = \Delta - (\Theta - \Gamma)$ \label{envSimpleEq:eq1}
  \item $\Gamma + (\Delta - \Gamma) = \Delta$ \label{envSimpleEq:eq2}
  \item $\Gamma - (\Gamma - \Delta) = \Delta$ \label{envSimpleEq:eq3}
\end{enumerate}
\end{lemma}

\section{A type system for the linear lambda-calculus}\label{Sec4}

With respect to our intended objective, i.e., providing a typing system for the linear
$\lambda$-calculus that properly accommodates $\lambda$-terms in de Bruijn notation,
the system of Figure~\ref{tsfllc} has still some defects. A first weakness is that the linearity of the typed
$\lambda$-terms is not guaranteed
by the typing system \textit{per se}, but by an external disjointness condition (the proviso of
Rule~ (\textsc{app})). Another weakness is that the system does not fit well with a backward-chaining
interpretation (which is useful for solving type inhabitance questions). The problem again comes from
Rule~ (\textsc{app}): given a fragmentary context $\Theta$, there is an exponential number of ways of splitting
it into two disjoint fragmentary environments $\Gamma$ and $\Delta$ such that $\Gamma+\Delta=\Theta$.

A possible remedy for the above shortcomings is to consider another kind of typing judgment:
\begin{align}\label{nktj}
	\tytm{\Gamma}{t}{\alpha}{\Delta} \text{\hspace{1cm}where $\Gamma,\Delta\in\fenv$, $t\in\ltm$, and $\alpha\in\ty$}
\end{align}

Typing judgement such as (\ref{nktj}) have been introduced by Hodas and Miller in the context of linear logic
programming~\cite{Hodas-Miller}.
Intuitively they may be interpreted as follows: given the typing resources of $\Gamma$, term $t$ can be assigned
type $\alpha$, and the typing resources that have not been consumed by this type assignment are kept in
$\Delta$.

A typing system based on judgements akin to (\ref{nktj}) is given in Figure~\ref{tsllc}.
When a typing judgement, $\tytm{\Gamma}{t}{\alpha}{\Delta}$, is derivable according to this system,
we write $\der{\Sll} \tytm{\Gamma}{t}{\alpha}{\Delta}$.

\begin{figure}[ht]
\fbox{%
\begin{minipage}{.975\textwidth}
\mbox{}
\vspace{\parsep}
\begin{center}
$
\tytm{\Gamma}{c}{\tau(c)}{\Gamma }\;\;\mbox{(\textsc{const})}
\quad\quad
\tytm{\Gamma, \alpha}{0}{\alpha}{\Gamma, \bt}\;\;\mbox{(\textsc{var})}
$
\end{center}
\begin{center}
$
\infer[\mbox{(\textsc{weak$_1$})}]{
\tytm{\Gamma, \beta}{\mathrm{i}+1}{\alpha }{\Delta, \beta}}{
\tytm{\Gamma}{\mathrm{i}}{\alpha }{\Delta}}
\quad\quad
\infer[\mbox{(\textsc{weak$_2$})}]{
\tytm{\Gamma, \bt}{\mathrm{i}+1}{\alpha }{\Delta, \bt}}{
\tytm{\Gamma}{\mathrm{i}}{\alpha }{\Delta}}
$
\end{center}
\begin{center}
$
\infer[\mbox{(\textsc{app})}]{
\tytm{\Gamma}{\iapp{t}{u}}{\beta}{\Theta}}{
\tytm{\Gamma}{t}{\alpha \rightarrow \beta}{\Delta}
&
\tytm{\Delta}{u}{\alpha}{\Theta}}
\quad\quad
\infer[\mbox{(\textsc{abs})}]{
\tytm{\Gamma}{\iabs{t}}{\alpha \rightarrow \beta}{\Delta}}{
\tytm{\Gamma, \alpha}{t}{\beta}{\Delta, \bt}}
$
\end{center}
\vspace{\parsep}
\end{minipage}
}
\caption{Typing system for the linear $\lambda$-calculus}\label{tsllc}
\end{figure}

We now prove a few lemmas and propositions that establish the correctness of the above type system.
The first two lemmas give a technical content to the intuition that in a judgement $\tytm{\Gamma}{t}{\alpha}{\Delta}$,
the fragmentary environment $\Delta$ is made of the declarations that were not used when assigning the type $\alpha$ to the
term $t$.

\begin{lemma}\label{subeq}
Let $\Gamma, \Delta \in \fenv$, $t\in\ltm$, and $\alpha\in\ty$. If $\der{\Sll}\tytm{\Gamma}{t}{\alpha}{\Delta}$
then $\Delta\sqsubseteq\Gamma$. 
\begin{proof}
The proof proceeds by a straightforward induction over the derivation of the typing judgement
$\der{\Sll}\tytm{\Gamma}{t}{\alpha}{\Delta}$.
\end{proof} 
\end{lemma}

\begin{lemma}\label{quasi-linearity-aux}
Let $\Gamma, \Delta\in\fenv$, $t\in\ltm$, and $\alpha\in\ty$ be such that
$\der{\Sll}\tytm{\Gamma}{t}{\alpha}{\Delta}$.
For every $i$ such that $0\leq i <|\Gamma|$, if $\Gamma(i)\in\ty$ and $\Delta(i)=\bt$
then $\occ{i}{t} = 1$, otherwise
$\Gamma(i)=\Delta(i)$ and $\occ{i}{t} = 0$.
\begin{proof}
The proof proceeds by induction over the derivation of the typing judgement $\der{\Sll}\tytm{\Gamma}{t}{\alpha}{\Delta}$.
All the cases are straightforward, except when $t$ is an application.
In this case, there exist $\Theta \in\fenv$, $t_1, t_2\in\ltm$, and
$\beta\in\ty$ such that $t=\iapp{t_1}{t_2}$, $\der{\Sll}\tytm{\Gamma}{t_1}{\beta\rightarrow\alpha}{\Theta}$,
and $\der{\Sll}\tytm{\Theta}{t_2}{\beta}{\Delta}$.
Let $0\leq i <|\Gamma|$, $\Gamma(i)\in\ty$ and $\Delta(i)=\bt$.
By Lemma~\ref{subeq}, we have $\Delta\sqsubseteq\Theta\sqsubseteq\Gamma$.
Consequently, either $\Theta(i)=\bt$ or $\Theta(i)=\Gamma(i)$.
In the first case, by the induction hypothesis, we have that $\occ{i}{t_1} = 1$ and $\occ{i}{t_2} =0$. 
In the second case, we have $\occ{i}{t_1} = 0$ and $\occ{i}{t_2} =1$.
Therefore, in both cases, $\occ{i}{t} = \occ{i}{t_1} + \occ{i}{t_2} = 1$.
Conversely, suppose that $\Gamma(i)\not\in\ty$ or $\Delta(i)\not=\bt$.
By Lemma~\ref{subeq}, we must have that $\Gamma(i)=\Theta(i)=\Delta(i)$.
Hence, by the induction hypothesis, $\occ{i}{t_1} = \occ{i}{t_2} =0$. Therefore, $\occ{i}{t} = 0$.
\end{proof}
\end{lemma}

Using Lemma~\ref{quasi-linearity-aux}, one can establish that the typing system of Figure~\ref{tsllc} allows only
quasi-linear $\lambda$-terms to be typed.

\begin{proposition}\label{quasi-linearity}
Let $\Gamma, \Delta\in\fenv$, $t\in\ltm$, and $\alpha\in\ty$ be such that
$\der{\Sll}\tytm{\Gamma}{t}{\alpha}{\Delta}$. Then, $t$ is quasi-linear.
\begin{proof}
The first condition of Definition~\ref{lltdef} can be established by induction on the derivation of the typing judgement
$\der{\Sll}\tytm{\Gamma}{t}{\alpha}{\Delta}$. When $t$ is a constant or a variable, the condition is vacuously satisfied.
When $t$ is an application, the induction is straightforward. In the case of an abstraction,
there exist $t_1\in\ltm$, and $\alpha_1, \alpha_2\in\ty$ such that
$t=\iabs{t_1}$, $\alpha=\alpha_1\rightarrow\alpha_2$, and $\der{\Sll}\tytm{\Gamma, \alpha_1}{t_1}{\alpha_2}{\Delta, \bt}$.
Then, by the induction hypothesis, $t_1$ satisfies the condition, and so does $\iabs{t_1}$ because $\occ{0}{t_1} = 1$
by Lemma~\ref{quasi-linearity-aux}.

The second condition of Definition~\ref{lltdef} follows from Lemma~\ref{quasi-linearity-aux}.
\end{proof}
\end{proposition}

We now prove that the typing system of Figure~\ref{tsllc} is correct in the sense that if it assigns a type $\alpha$ to a term $t$, then $t$ is indeed a simply typed $\lambda$-term of type $\alpha$.

\begin{proposition}\label{ll2st}
Let $\Gamma, \Delta\in\fenv$, $t\in\ltm$, and $\alpha\in\ty$ be such that $\der{\Sll}\tytm{\Gamma}{t}{\alpha}{\Delta}$.
Then, there exists $\Theta\in\env$ such that $\Gamma\sqsubseteq\Theta$ and
$\der{\Sst} \tseq{\Theta}{t}{\alpha}$.
\begin{proof}
The proof proceeds by induction over the derivation of the typing judgement $\der{\Sll}\tytm{\Gamma}{t}{\alpha}{\Delta}$.
In order to make the induction work, we prove the more general statement that, for every $\Theta\in\env$,
if $\Gamma\sqsubseteq\Theta$, then $\der{\Sst} \tseq{\Theta}{t}{\alpha}$. The induction is straightforward, except in the case
of an application, where Lemma \ref{subeq} is needed to allow the induction hypothesis to be applied to the second premise of
the rule.
\end{proof}
\end{proposition}

We now want to give a converse to Propositions~\ref{quasi-linearity} and~\ref{ll2st}. To this end, we first define what it means
for a fragmentary environment to be compatible with a term.
Let $\Gamma \in \fenv$ be a fragmentary environment, and let $t\in\ltm$ be a $\lambda$-term. We say that $\Gamma$ is compatible
with $t$ if and only if, for all $i\in\fv{t}$, $i<|\Gamma|$ and $\Gamma(i)\not=\bt$.

\begin{proposition}\label{st2ll}
Let $\Gamma\in\env$, $t\in\ltm$, and $\alpha\in\ty$ be such that $\der{\Sst}\tseq{\Gamma}{t}{\alpha}$.
If $t$ is quasi-linear, then there exists $\Delta\in\fenv$ such that
$\der{\Sll}\tytm{\Gamma}{t}{\alpha}{\Delta}$.
\begin{proof}
The proof proceeds by induction over the derivation of the typing judgement $\der{\Sst} \tseq{\Gamma}{t}{\alpha}$.
All the cases are straightforward except for the case of an application for which the induction hypothesis needs to be
strengthened. Accordingly, we prove the following more general statement: for every $\Theta \in \fenv$ compatible with $t$,
if $\Theta\sqsubseteq\Gamma$, then there exists $\Delta\in\fenv$ such that $\der{\Sll}\tytm{\Theta}{t}{\alpha}{\Delta}$.

Let us focus on the non-easy case, that is, when $t$ is an application. In this case, there exist $t_1, t_2\in\ltm$, and
$\beta\in\ty$ such that $t=\iapp{t_1}{t_2}$, $\der{\Sst}\tseq{\Gamma}{t_1}{\beta\rightarrow\alpha}$,
and $\der{\Sst}\tseq{\Gamma}{t_2}{\beta}$. Let $\Theta \in \fenv$ be compatible with $t$ and such that $\Theta\sqsubseteq\Gamma$.
Since $\Theta$ is compatible with $t$, it is a fortiori compatible with $t_1$. Hence, by applying the induction hypothesis on the
first premise, there exists $\Lambda \in \fenv$ such that $\der{\Sll}\tytm{\Theta}{t_1}{\beta\rightarrow\alpha}{\Lambda}$.
Now, in order to apply the induction hypothesis on the second premise with $\Lambda$, we must prove that
$\Lambda\sqsubseteq\Gamma$ and that $\Lambda$ is compatible with $t_2$. The first assertion follows from Lemma~\ref{subeq}.
As for the second one, let $i\in\fv{t_2}$. A fortiori $i\in\fv{t}$, and since $\Theta$ is compatible with $t$,
$\Theta(i) \not=\bt$. Furthermore, $t$ being quasi-linear, $\occ{i}{t}=1$, which implies that $\occ{i}{t_1}=0$
because $i\in\fv{t_2}$. Hence, by Lemma~\ref{quasi-linearity-aux}, $\Lambda(i) = \Theta(i) \not= \bt$.
This shows that $\Lambda$ is compatible with $t_2$. We may then apply the induction hypothesis on the second premise,
which yields that there exists $\Delta \in \fenv$ such that $\der{\Sll}\tytm{\Lambda}{t_2}{\beta}{\Delta}$.
We may then conclude that $\der{\Sll}\tytm{\Theta}{\iapp{t_1}{t_2}}{\alpha}{\Delta}$.
\end{proof}
\end{proposition}

Using the preceding three propositions, we can show that the typing system of Figure~\ref{tsllc} can be used to characterize
exactly the linear simply-typed $\lambda$-terms. To state this property we first need to introduce a notion of minimality
of an environment with respect to a term. Let $\Gamma \in \env$ be an environment, and let $t\in\ltm$ be a $\lambda$-term.
We say that $\Gamma$ is minimal with respect to $t$ if and only if $(|\Gamma|-1)\in\fv{t}$ and for all $i\in\fv{t}$, $i<|\Gamma|$.

\begin{proposition}
Let $\Gamma\in\env$, $t\in\ltm$, and $\alpha\in\ty$.
Then, $\der{\Sll}\tytm{\Gamma}{t}{\alpha}{\minenv}$ if and only if
$t$ is a linear $\lambda$-term, $\Gamma$ is minimal with respect to $t$, and $\der{\Sst} \tseq{\Gamma}{t}{\alpha}$.
\begin{proof}
Suppose that $\der{\Sll}\tytm{\Gamma}{t}{\alpha}{\minenv}$. By proposition \ref{quasi-linearity}, we have that $t$
is quasi-linear. In addition, by Lemma \ref{quasi-linearity-aux}, for every $i$ such that $0\leq i <|\Gamma|$, we have
that $i\in\fv{t}$. This implies that $t$ is linear and that $\Gamma$ is minimal with respect to $t$. Finally,
by proposition~\ref{ll2st}, we have that $\der{\Sst} \tseq{\Gamma}{t}{\alpha}$.

Conversely, suppose that $t$ is a linear $\lambda$-term, that $\Gamma$ is minimal with respect to $t$, and that
$\der{\Sst} \tseq{\Gamma}{t}{\alpha}$.
By proposition~\ref{st2ll}, there exists $\Delta\in\fenv$ such that $\der{\Sll}\tytm{\Gamma}{t}{\alpha}{\Delta}$.
Since $\Gamma$ is minimal with respect to $t$, we have that $(|\Gamma|-1)\in\fv{t}$. Then, since $t$ is linear, for every $i$
such that $0\leq i<|\Gamma|$, $\occ{i}{t} = 1$. This implies, by Lemma~\ref{quasi-linearity-aux}, that $\Delta=\minenv$.
\end{proof}
\end{proposition}

We end this section by stating two technical lemmas that will be needed later on.

\begin{lemma}\label{lemma:envdefined}
Let $\Gamma, \Delta, \Theta \in \fenv$, such that $\Delta \sqsubseteq \Gamma$. Then:
\begin{enumerate}[label={\rm (\alph*)}, labelsep*=2ex, itemsep=0pt]
  \item If\/ $\Gamma + \Theta$ is defined, then $\Delta + \Theta$ is defined. \label{envdefined:add}
  \item If\/ $\Theta - \Gamma$ is defined, then $\Theta - \Delta$ is defined. \label{envdefined:sub}
\end{enumerate}
\begin{proof}
(a)
We proceed componentwisely. Let $i<|\Gamma|$. Either $\Gamma(i)\in\ty$ or $\Gamma(i)=\bt$.
In the first case, since $\Gamma + \Theta$ is defined, we have $\Theta(i)=\bt$.
Consequently, $\Delta(i) + \Theta(i)$ is defined. In the second case, since $\Delta \sqsubseteq \Gamma$, we have that
$\Delta(i) = \bt$, which implies that $\Delta(i) + \Theta(i)$ is defined.
(b) The proof is similar.
\end{proof}
\end{lemma}

\begin{lemma}\label{lemma:envarith}
Let $\Gamma, \Delta, \Theta \in \fenv$, $t \in \ltm$, and $\alpha \in \ty$ such that $\der{\Sll}\tytm{\Gamma}{t}{\alpha}{\Delta}$.
\begin{enumerate}[label={\rm (\alph*)}, labelsep*=2ex, itemsep=0pt]
  \item If\/ $\Gamma + \Theta$ is defined, then $\Delta + \Theta$ is defined and
    $\der{\Sll}\tytm{\Gamma + \Theta}{t}{\alpha}{\Delta + \Theta}$. \label{envarith:add}
  \item If\/ $\Delta - \Theta$ is defined, then $\Gamma - \Theta$ is defined and
    $\der{\Sll}\tytm{\Gamma - \Theta}{t}{\alpha}{\Delta - \Theta}$. \label{envarith:sub}
\end{enumerate}
\begin{proof}
By Lemma~\ref{subeq}, we have $\Delta \sqsubseteq \Gamma$.
Hence, for case \ref{envarith:add}, by Lemma~\ref{lemma:envdefined} \ref{envdefined:add}, $\Delta + \Theta$ is defined.
Similarly, for case \ref{envarith:sub}, by Lemma~\ref{lemma:envdefined} \ref{envdefined:sub}, $\Gamma - \Theta$ is defined.
Then, for both cases, the proof proceeds by a straightforward induction on the derivation of
$\der{\Sll}\tytm{\Gamma}{t}{\alpha}{\Delta}$.
\end{proof}
\end{lemma}

\section{Beta-reduction with de Bruijn indices}

When using the standard syntax with named variables, it is easy to show that linearity is stable by $\beta$-reduction.
Consequently, the notion of $\beta$-reduction for the linear $\lambda$-calculus is the same as usual.
Nevertheless, it is worth reviewing how beta-reduction is implemented using de Bruijn indices because it involves a few
subtleties related to the updating of the indices when a substitution is made within a $\lambda$-abstraction.

To illustrate this updating issue, let us consider the following generalized weakening rule, which is admissible in
the simply typed $\lambda$-calculus:
$$
\infer{
\tseq{\Gamma,\Delta}{t}{\alpha}
}{
\tseq{\Gamma}{t}{\alpha}
}
$$
While this rule is perfectly correct if $t$ is specified using the standard syntax of the $\lambda$-calculus, it is not if $t$
is written in de Bruijn notation. The problem is that the values of the indices that occur free in $t$ must be updated in the
conclusion of the rule.
Indeed, in the premise of the rule, those indices are pointing to positions in the environment $\Gamma$.
Accordingly, in the conclusion of the rule, in order to have them pointing to the same positions in $\Gamma$,
their values must be incremented by the length of $\Delta$.

The operation of updating the free indices of a term $t$ is called \textit{lifting}.
It is denoted by $\liftt{t}{k}{j}$, where $t$ is the term whose free indices must be updated, $k$ is the increment,
and $j$ is the number of nested $\lambda$-abstractions within which $t$ occurs.

\begin{definition}
 Let $t \in \ltm$ and, $i, k \in \mathbb{N}$. The lifting operation $\liftt{t}{k}{i}$ is inductively defined as follows:
\begin{enumerate}[label={\rm \roman*.}, labelsep*=2ex, parsep=0pt, itemsep=3pt plus 1pt minus 1pt]
\item $\liftt{c}{k}{i} = c$,\, for $c\in\cst$
\item $\liftt{j}{k}{i} =
	\begin{cases}
	 j & \text{if $j < i$} \\
	 j + k & \text{if $j \geq i$}
	\end{cases}$
\label{lift:var}
\item $\liftt{\iapp*{t_1}{t_2}}{k}{i} = \iapp{\liftt{t_1}{k}{i}}{\liftt{t_2}{k}{i}}$
\item $\liftt{\iabs*{t_1}}{k}{i} = \iabs{\liftt{t_1}{k}{i + 1}}$
\end{enumerate}
\end{definition}

Using de Bruijn indices instead of named variables, the usual operation of substituting a variable by a term is then specified
by the following definition.

\begin{definition}
Let $t, u \in \ltm$ and $i \in \mathbb{N}$. The substitution of index $i$ by term $u$ in term $t$,
in notation $\sbstt{t}{i}{u}$ is inductively defined as follows:
\begin{enumerate}[label={\rm \roman*.}, labelsep*=2ex, itemsep=0pt]
  \item $\sbstt{c}{i}{u} = c$,\, for $c\in\cst$
  \item $\sbstt{j}{i}{u} =
    \begin{cases}
       j & \text{if $j < i$} \\
       \liftt{u}{i}{0} & \text{if $j = i$} \\
       j - 1 & \text{if $j > i$}
    \end{cases}$
    \label{subst:var}
  \item $\sbstt{\iapp{t_1}{t_2}}{i}{u} = \iapp{\sbstt{t_1}{i}{u}}{\sbstt{t_2}{i}{u}}$
  \item $\sbstt{\iabs{t_1}}{i}{u} = \iabs{\sbstt{t_1}{i+1}{u}}$
    \label{subst:abs}
\end{enumerate}
\end{definition}

The operation of substitution is then used to define the relation of $\beta$-contraction.

\begin{definition} \label{def:beta}
Let $t, u \in \ltm$. The relation of $\beta$-contraction, in notation $t \rightarrow_\beta u$, is the smallest relation that obey
the following rules:
$$
\iapp{\iabs{t}}{u} \rightarrow_{\beta} \sbstt{t}{0}{u}
$$
$$
\infer{
\iapp{t}{v} \rightarrow_{\beta} \iapp{u}{v}
}{
t \rightarrow_{\beta} u
}
\quad\quad
\infer{
\iapp{v}{t} \rightarrow_{\beta} \iapp{v}{u}
}{
t \rightarrow_{\beta} u
}
\quad\quad
\infer{
\iabs{t} \rightarrow_{\beta} \iabs{u}
}{
t \rightarrow_{\beta} u
}
$$
\end{definition}

Finally, the relation of $\beta$-reduction is defined as the reflexive, transitive closure of the relation of $\beta$-contraction.

\section{Subject reduction}

A key property that a typing system is expected to satisfy with respect to a given relation of reduction is the subject
reduction property. This property ensures that typing is stable under evaluation. In the case of the typing system we
have defined in Section \ref{Sec4}, this property is stated as follows:
\begin{quote}
\textit{Let $t\in \ltm$, $\alpha \in \ty$, and $\Gamma, \Delta \in \fenv$ be such that
$\der{\Sll}\tytm{\Gamma}{t}{\alpha}{\Delta}$. If $u\in \ltm$ is such that
$t\rightarrow_\beta u$, then $\der{\Sll}\tytm{\Gamma}{u}{\alpha}{\Delta}$.}
\end{quote}

\noindent
This property, can be proved by induction on the definition of the relation of $\beta$-contraction as specified in
Definition~\ref{def:beta}. To do so, we first need to establish a substitution lemma that corresponds to the base
case of the induction.
To this end, we start by stating and proving a lemma that concerns the substitution of a term for a variable
that does not occur in the term over which the substitution is performed.

\begin{lemma}\label{lemma:l3}
Let $\Gamma_1, \Gamma_2, \Delta_1, \Delta_2 \in \fenv$, $\omega \in \qty$, $t, u \in \ltm$, $\alpha \in \ty$,
and $i \in \mathbb{N}$ be such that
$|\Gamma_2| = |\Delta_2| = i$ and
$\der{\Sll}\tytm{\Gamma_1, \omega, \Gamma_2}{t}{\alpha}{\Delta_1, \omega, \Delta_2}$.
Then $\der{\Sll}\tytm{\Gamma_1, \Gamma_2}{\sbstt{t}{i}{u}}{\alpha}{\Delta_1, \Delta_2}$.
\begin{proof}
The proof proceeds by induction on $t$.
We focus on the case where $t$ is an index, the other cases being straightforward.

Let $t = j$. We have that $\der{\Sll}\tytm{\Gamma_1, \omega, \Gamma_2}{j}{\alpha}{\Delta_1, \omega, \Delta_2}$
is obtained by a derivation consisting of Axiom (\textsc{var}) followed by a chain of $j$ weakening rules ($\textsc{weak}_1$ or
$\textsc{weak}_2$).
We distinguish between three cases, depending on the value of $j$.

\begin{itemize}
\item $j < i$ and $\sbstt{j}{i}{u} = j$.\quad
In this case, the derivation of the typing judgement is as follows:
\begin{align*}
\infer*[\quad\mbox{(\textsc{weak})}]{
\tytm{\Gamma_1, \omega, \Gamma_2}{j}{\alpha}{\Delta_1, \omega, \Delta_2}
}{
\tytm{\Gamma_1, \omega, \Theta,\alpha}{0}{\alpha}{\Delta_1, \omega, \Theta,\bt}\;\; \mbox{(\textsc{var})}
}
\end{align*}
where $\Gamma_1=\Delta_1$, $\Theta,\alpha$ is a prefix of $\Gamma_2$, and $\Theta, \bt$ is a prefix of $\Delta_2$.
Hence, we obtain the expected result from the following derivation:
\begin{align*}
\infer*[\quad\mbox{(\textsc{weak})}]{
\tytm{\Gamma_1, \Gamma_2}{j}{\alpha}{\Delta_1, \Delta_2}
}{
\tytm{\Gamma_1, \Theta,\alpha}{0}{\alpha}{\Delta_1, \Theta,\bt}\;\; \mbox{(\textsc{var})}
}
\end{align*}

\item $j=i$.\quad
This case is not possible because the derivation of the typing judgement would start with
$\tytm{\Gamma_1, \omega}{0}{\alpha}{\Delta_1, \omega}$ which is cannot be a valid instance of Axiom ($\textsc{var}$). 

\item $j > i$ and $\sbstt{j}{i}{u} = j - 1$.\quad
In this case, the derivation of the typing judgement is the following:
\begin{align*}
\infer*[\quad\mbox{(\textsc{weak})}]{
\tytm{\Gamma_1, \omega, \Gamma_2}{j}{\alpha}{\Delta_1, \omega, \Delta_2}
}{
\infer[\mbox{(\textsc{weak})}]{
\tytm{\Gamma_1, \omega}{j-i}{\alpha}{\Delta_1, \omega}
}{
\infer*[\mbox{\quad(\textsc{weak})}]{
\tytm{\Gamma_1}{(j-1)-i}{\alpha}{\Delta_1}
}{
\tytm{\Theta,\alpha}{0}{\alpha}{\Theta,\bt}\;\; \mbox{(\textsc{var})}
}}}
\end{align*}
where $\Theta,\alpha$ is a prefix of $\Gamma_1$, and $\Theta, \bt$ is a prefix of $\Delta_1$.
We then obtain the expected result as follows:
\begin{align*}
\infer*[\quad\mbox{(\textsc{weak})}]{
\tytm{\Gamma_1, \omega, \Gamma_2}{j-1}{\alpha}{\Delta_1, \omega, \Delta_2}
}{
\infer*[\mbox{\quad(\textsc{weak})}]{
\tytm{\Gamma_1}{(j-1)-i}{\alpha}{\Delta_1}
}{
\tytm{\Theta,\alpha}{0}{\alpha}{\Theta,\bt}\;\; \mbox{(\textsc{var})}
}}
\end{align*}
\end{itemize}
\end{proof}
\end{lemma}

The next lemma, which is rather technical, is about the operation of lifting.

\begin{lemma}\label{lemma:philippe}
Let $\Gamma_1, \Gamma_2, \Delta_1, \Delta_2, \Theta \in \fenv$, $t \in \ltm$, $\alpha \in \ty$, and $k, i \in \mathbb{N}$ be
such that $|\Gamma_2| = |\Delta_2| = i$, $|\Theta| = k$, and $\der{\Sll}\tytm{\Gamma_1, \Gamma_2}{t}{\alpha}{\Delta_1, \Delta_2}$.
Then $\der{\Sll}\tytm{\Gamma_1, \Theta, \Gamma_2}{\liftt{t}{k}{i}}{\alpha}{\Delta_1, \Theta, \Delta_2}$.
\begin{proof}
The proof proceeds by induction on $t$.
\begin{itemize}
\item $t = c$.\quad
We have that $\der{\Sll}\tytm{\Gamma_1, \Gamma_2}{c}{\alpha}{\Delta_1, \Delta_2}$. Since $\Gamma_2$ and $\Delta_2$ have the same
length, we have that $\Gamma_1 = \Delta_1$ and $\Gamma_2 = \Delta_2$. Therefore,
$\der{\Sll}\tytm{\Gamma_1, \Theta, \Gamma_2}{c}{\alpha}{\Gamma_1, \Theta, \Gamma_2}$, which yields the expected result since
$\liftt{c}{k}{i} = c$.

\item If $t = j$.\quad
We have that $\der{\Sll}\tytm{\Gamma_1, \Gamma_2}{j}{\alpha}{\Delta_1, \Delta_2}$ and $(\Gamma_1, \Gamma_2)(j)=\alpha$.
We distinguish between two cases, depending on the values of $j$, the \textit{lift} function will expand differently.
\begin{itemize}
\item If $j < i$ and $\liftt{j}{k}{i} = j$.\quad
Since $|\Gamma_2| = i$, we have $(\Gamma_1, \Gamma_2)(j)=\Gamma_2(j) = \alpha$.
This implies that $\der{\Sll}\tytm{\Gamma_1, \Theta, \Gamma_2}{j}{\alpha}{\Delta_1, \Theta, \Delta_2}$ holds, because the
index $j$ corresponds to the same position in both $\{\Gamma_1, \Gamma_2\}$ and $\{\Gamma_1, \Theta, \Gamma_2\}$.
\item $j \ge i$ and $\liftt{j}{k}{i} = j + k$.\quad
Since $|\Gamma_2| = i$, we have $\Gamma_1(j - i) = (\Gamma_1, \Gamma_2)(j) = \alpha$.
Then, because $|\Theta| = k$, $\Gamma_1(j - i) = (\Gamma_1,\theta)(j + k - i) = (\Gamma_1,\theta,\Gamma_2)(j+k)$.
Hence $\der{\Sll}\tytm{\Gamma_1, \Theta, \Gamma_2}{j + k}{\alpha}{\Delta_1, \Theta, \Delta_2}$.
\end{itemize}

\item $t = \iapp{t_1}{t_2}$.\quad
we have $\der{\Sll}\tytm{\Gamma_1, \Gamma_2}{\iapp{t_1}{t_2}}{\alpha}{\Delta_1, \Delta_2}$. This typing judgement is
obtained by a derivation of the following form:
\begin{align*}
\infer[\mbox{(\textsc{app})}]{
\tytm{\Gamma_1, \Gamma_2}{\iapp{t_1}{t_2}}{\alpha}{\Delta_1, \Delta_2}
}{
\infer*{\tytm{\Gamma_1, \Gamma_2}{t_1}{\beta \rightarrow \alpha}{\Xi_1, \Xi_2}}{}
&
\infer*{\tytm{\Xi_1, \Xi_2}{t_2}{\beta}{\Delta_1, \Delta_2}}{}
}
\end{align*}
for some $\Xi_1, \Xi_2 \in \fenv$ and $\beta \in \ty$ such that $|\Xi_2|=i$.
Then, by induction hypothesis, we obtained the following derivation:
\begin{align*}
\infer[\mbox{(\textsc{app})}]{
\tytm{\Gamma_1, \Theta, \Gamma_2}{\iapp{\liftt{t_1}{k}{i}}{\liftt{t_2}{k}{i}}}{\alpha}{\Delta_1, \Theta, \Delta_2}
}{
\infer*{\tytm{\Gamma_1, \Theta, \Gamma_2}{\liftt{t_1}{k}{i}}{\beta \rightarrow \alpha}{\Xi_1, \Theta, \Xi_2}}{}
&
\infer*{\tytm{\Xi_1, \Theta, \Xi_2}{\liftt{t_2}{k}{i}}{\beta}{\Delta_1, \Theta, \Delta_2}}{}
}
\end{align*}
This allows us to conclude because
$\liftt{\iapp{t_1}{t_2}}{k}{i} = \iapp{\liftt{t_1}{k}{i}}{\liftt{t_2}{k}{i}}$.

\item $t = \iabs{t_1}$. Then, there exists $\alpha_1, \alpha_2 \in \ty$ such that $\alpha = (\alpha_1\rightarrow\alpha_2)$ and
$\der{\Sll}\tytm{\Gamma_1, \Gamma_2}{\iabs{t_1}}{\alpha_1\rightarrow\alpha_2}{\Delta_1, \Delta_2}$. This typing judgement is
obtained by the following derivation: 
\begin{align*}
\infer[\mbox{(\textsc{abs})}]{
\tytm{\Gamma_1, \Gamma_2}{\iabs{t_1}}{\alpha_1 \rightarrow \alpha_2}{\Delta_1, \Delta_2}
}{
\infer*{\tytm{\Gamma_1, \Gamma_2, \alpha_1}{t_1}{\alpha_2}{\Delta_1, \Delta_2, \bt}}{}
}
\end{align*}
Then, by induction hypothesis, there exists a derivation of the following form:
\begin{align*}
\infer[\mbox{(\textsc{abs})}]{
\tytm{\Gamma_1, \Theta, \Gamma_2}{\iabs{\liftt{t_1}{k}{i + 1}}}{\alpha_1 \rightarrow \alpha_2}{\Delta_1, \Theta, \Delta_2}
}{
\infer*{\tytm{\Gamma_1, \Theta, \Gamma_2, \alpha_1}{\liftt{t_1}{k}{i + 1}}{\alpha_2}{\Delta_1, \Theta, \Delta_2, \bt}}{}
}
\end{align*}
This yields the expected result since $\liftt{\iabs*{t_1}}{k}{i} = \iabs{\liftt{t_1}{k}{i + 1}}$.
\end{itemize}
\end{proof}
\end{lemma}

We are now in a position to state and prove the substitution lemma.

\begin{lemma} \label{lemma:sr}
Let $\Gamma_1, \Gamma_2, \Delta_1, \Delta_2, \Theta \in \fenv$, $t, u \in \ltm$, $\alpha, \beta \in \ty$,
and $i \in \mathbb{N}$ be such that $|\Gamma_2| = |\Delta_2| = i$,
$\der{\Sll}\tytm{\Gamma_1, \beta, \Gamma_2}{t}{\alpha}{\Delta_1, \bt, \Delta_2}$, and
$\der{\Sll}\tytm{\Delta_1}{u}{\beta}{\Theta}$.
Then $\der{\Sll}\tytm{\Gamma_1, \Gamma_2}{\sbstt{t}{i}{u}}{\alpha}{\Theta, \Delta_2}$.
\begin{proof}
The proof proceeds by induction on $t$.  Note that by Lemma~\ref{quasi-linearity-aux}, $\occ{i}{t} = 1$.
\begin{itemize}
\item $t = c$.\quad This case is impossible because, for any constant $c$, $\occ{i}{c} = 0$.

\item $t = j$.\quad Since $\occ{i}{t} = 1$, we must have $i=j$.  Accordingly, we have $\alpha = \beta$, $\Gamma_2 = \Delta_2$,
and $\sbstt{j}{i}{u} = \liftt{u}{i}{0}$.  Now, by hypothesis, we have $\der{\Sll}\tytm{\Delta_1}{u}{\alpha}{\Theta}$. Hence,
by Lemma~\ref{lemma:philippe}, we obtain $\der{\Sll}\tytm{\Gamma_1, \Gamma_2}{\sbstt{j}{i}{u}}{\alpha}{\Theta, \Delta_2}$.

\item $t = \iapp{t_1}{t_2}$.\quad We distinguish between two cases depending on whether the index $i$ occurs in $t_1$ or $t_2$.
\begin{itemize}
\item $\occ{i}{t_1} = 1$ and $\occ{i}{t_2} = 0$.\quad
In this case, $\der{\Sll}{\tytm{\Gamma_1, \beta, \Gamma_2}{\iapp{t_1}{t_2}}{\alpha}{\Delta_1, \bt, \Delta_2}}$ is
obtained by a derivation whose last rule is of the following form:
\begin{align}\label{der:app1}
\infer[\mbox{(\textsc{app})}]{
\tytm{\Gamma_1, \beta, \Gamma_2}{\iapp{t_1}{t_2}}{\alpha}{\Delta_1, \bt, \Delta_2}
}{
\tytm{\Gamma_1, \beta, \Gamma_2}{t_1}{\gamma \rightarrow \alpha}{\Xi_1, \bt, \Xi_2}
&
\tytm{\Xi_1, \bt, \Xi_2}{t_2}{\gamma}{\Delta_1, \bt, \Delta_2}
}
\end{align}
for some $\gamma\in\ty$ and some $\Xi_1, \Xi_2\in\fenv$ such that $|\Xi_2|=i$.
By Lemma~\ref{subeq}, $\Delta_1\sqsubseteq\Xi_1$. Accordingly, $\Xi_1-\Delta_1$ is defined, and so is
$\Delta_1+(\Xi_1-\Delta_1)$.  Hence, by Lemma~\ref{lemma:envarith}\ref{envarith:add}, we have that
$\der{\Sll}\tytm{\Delta_1+(\Xi_1-\Delta_1)}{u}{\beta}{\Theta+(\Xi_1-\Delta_1)}$.
By Lemma~\ref{lemma:envSimpleEq}\ref{envSimpleEq:eq2}, this can be rewritten as
$\der{\Sll}\tytm{\Xi_1}{u}{\beta}{\Theta+(\Xi_1-\Delta_1)}$.
Then, by applying the induction hypothesis on the first premise of (\ref{der:app1}), we obtain that the following judgement
is derivable:
\begin{align}\label{app1:pr1}
\der{\Sll}\tytm{\Gamma_1, \Gamma_2}{\sbstt{t_1}{i}{u}}{\gamma \rightarrow \alpha}{\Theta + (\Xi_1 - \Delta_1), \Xi_2}
\end{align}
Now, from the second premise of (\ref{der:app1}), by Lemma~\ref{lemma:l3},
$\der{\Sll}{\tytm{\Xi_1, \Xi_2}{\sbstt{t_2}{i}{u}}{\gamma}{\Delta_1, \Delta_2}}$.
From this, by Lemma~\ref{lemma:envarith}\ref{envarith:sub},
$\der{\Sll}{\tytm{\Xi_1-(\Delta_1 - \Theta), \Xi_2}{\sbstt{t_2}{i}{u}}{\gamma}{\Delta_1-(\Delta_1 - \Theta), \Delta_2}}$.
By Lemma~\ref{lemma:envSimpleEq}, \ref{envSimpleEq:eq1} and \ref{envSimpleEq:eq3}, this last judgement may be rewritten as
follows:
\begin{align}\label{app1:pr2}
  \der{\Sll}{\tytm{\Theta + (\Xi_1-\Delta_1), \Xi_2}{\sbstt{t_2}{i}{u}}{\gamma}{\Theta, \Delta_2}}
\end{align}
Finally, since $\sbstt{\iapp{t_1}{t_2}}{i}{u} = \iapp{\sbstt{t_1}{i}{u}}{\sbstt{t_2}{i}{u}}$,
we obtained the expected result from \ref{app1:pr1} and \ref{app1:pr2} by applying Rule (\textsc{app}):
\begin{align*}
\infer{
\tytm{\Gamma_1, \Gamma_2}{\iapp{\sbstt{t_1}{i}{u}}{\sbstt{t_2}{i}{u}}}{\alpha}{\Theta, \Delta_2}
}{
\tytm{\Gamma_1, \Gamma_2}{\sbstt{t_1}{i}{u}}{\gamma \rightarrow \alpha}{\Theta + (\Xi_1 - \Delta_1), \Xi_2}
&
\tytm{\Theta + (\Xi_1-\Delta_1), \Xi_2}{\sbstt{t_2}{i}{u}}{\gamma}{\Theta, \Delta_2}
}
\end{align*}

\item $\occ{i}{t_1} = 0$ and $\occ{i}{t_2} = 1$.\quad
In this second case, $\der{\Sll}{\tytm{\Gamma_1, \beta, \Gamma_2}{\iapp{t_1}{t_2}}{\alpha}{\Delta_1, \bt, \Delta_2}}$ is
obtained by a derivation whose last rule is as follows:
\begin{align}\label{der:app2}
\infer[\mbox{(\textsc{app})}]{
\tytm{\Gamma_1, \beta, \Gamma_2}{\iapp{t_1}{t_2}}{\alpha}{\Delta_1, \bt, \Delta_2}
}{
\tytm{\Gamma_1, \beta, \Gamma_2}{t_1}{\gamma \rightarrow \alpha}{\Xi_1, \beta, \Xi_2}
&
\tytm{\Xi_1, \beta, \Xi_2}{t_2}{\gamma}{\Delta_1, \bt, \Delta_2}
}
\end{align}
for some $\gamma\in\ty$ and some $\Xi_1, \Xi_2\in\fenv$ such that $|\Xi_2|=i$.
On the one hand, by applying Lemma~\ref{lemma:l3} to the first premise of \ref{der:app2}, we have that
$\der{\Sll}{\tytm{\Gamma_1, \Gamma_2}{\sbstt{t_1}{i}{u}}{\gamma \rightarrow \alpha}{\Xi_1, \Xi_2}}$.
On the other hand, by applying the induction hypothesis to the second premise of \ref{der:app2}, we obtain
$\der{\Sll}{\tytm{\Xi_1, \Xi_2}{\sbstt{t_2}{i}{u}}{\gamma}{\Theta, \Delta_2}}$. Then, the expected result follows from these
two judgements by an application of Rule (\textsc{app}).
  
\end{itemize}

\item $t = \iabs{t_1}$.\quad
In this case, $\der{\Sll}\tytm{\Gamma_1, \beta, \Gamma_2}{\iabs{t_1}}{\alpha}{\Delta_1, \bt, \Delta_2}$ is obtained
from a derivation whose last rule is Rule (\textsc{abs}):
\begin{align*}
\infer[\mbox{(\textsc{abs})}]{
\tytm{\Gamma_1, \beta, \Gamma_2}{\iabs{t_1}}{\alpha_1 \rightarrow \alpha_2}{\Delta_1, \bt, \Delta_2}
}{
\tytm{\Gamma_1, \beta, \Gamma_2, \gamma}{t_1}{\alpha_2}{\Delta_1, \bt, \Delta_2, \bt}
}
\end{align*}
for some $\alpha_1, \alpha_2 \in \ty$ such that $\alpha = \alpha_1 \rightarrow \alpha_2$.
Since $\sbstt{\iabs{t_1}}{i}{u} = \iabs{\sbstt{t_1}{i + 1}{u}}$, this can be obtained by showing that its premise holds:we get the
expected result by a straightforward application of the induction hypothesis:
\begin{align*}
\infer[\mbox{(\textsc{abs})}]{
\tytm{\Gamma_1, \Gamma_2}{\iabs{\sbstt{t_1}{i + 1}{u}}}{\alpha_1 \rightarrow \alpha_2}{\Theta, \Delta_2}
}{
\tytm{\Gamma_1, \Gamma_2, \alpha_1}{\sbstt{t_1}{i + 1}{u}}{\alpha_2}{\Theta, \Delta_2, \bt}
}
\end{align*}
\end{itemize}
\end{proof}
\end{lemma}

We now establish the subject-reduction property.

\begin{proposition}
Let $\Gamma, \Delta \in \fenv$, $t\in\ltm$, and $\alpha \in \ty$ be such that $\der{\Sll}\tytm{\Gamma}{t}{\alpha}{\Delta}$.
Then, for every $u\in\ltm$ such that $t \rightarrow_{\beta} u$, we have that $\der{\Sll}\tytm{\Gamma}{u}{\alpha}{\Delta}$.
\begin{proof}
The proof proceeds by induction on the derivation of $t \rightarrow_{\beta} u$.
The base case corresponds to the $\beta$-contraction of a reducible expression of the form $\iapp{\iabs{t_1}}{t_2}$,
and the inductive cases correspond to the congruence rules of Definition~\ref{def:beta}.
The latter being straightforward, we concentrate on the former. 

Let $t = \iapp{\iabs{t_1}}{t_2}$ and $u = \sbstt{t_1}{0}{t_2}$.
Since $\der{\Sll}\tytm{\Gamma}{\iapp{\iabs{t_1}}{t_2}}{\alpha}{\Delta}$, there exists a typing derivation that ends as follows:
\begin{align}
\infer[\mbox{(\textsc{app})}]{
\tytm{\Gamma}{\iapp{\iabs{t_1}}{t_2}}{\alpha}{\Delta}
}{
\infer[\mbox{(\textsc{abs})}]{
\tytm{\Gamma}{\iabs{t_1}}{\beta \rightarrow \alpha}{\Theta}
}{
\infer*{\tytm{\Gamma, \beta}{t_1}{\alpha}{\Theta, \bt}}{}
}
&
\infer*{\tytm{\Theta}{t_2}{\beta}{\Delta}}{}
}
\end{align}
for some $\beta\in\ty$ and $\Theta \in \fenv$.
Therefore, we have that $\der{\Sll}\tytm{\Gamma, \beta}{t_1}{\alpha}{\Theta, \bt}$
and $\der{\Sll}\tytm{\Theta}{t_2}{\beta}{\Delta}$.
Then, by Lemma~\ref{lemma:sr},
$\der{\Sll}\tytm{\Gamma}{\sbstt{t_1}{0}{t_2}}{\alpha}{\Delta}$. 
\end{proof}
\end{proposition}

\section{Conclusion}

We have developed a typing system for the linear $\lambda$-calculus in de Bruijn notation, and we have investigated several
of its properties. In particular, we have proved that it satisfies the subject reduction property.

Our typing system is the result of combining three key elements: Hodas' and Miller's model of resource consumption
\cite{Hodas-Miller}, de Bruijn indices\cite{deBruijn:indices}, and fragmentary environments. Hodas'
and Miller's model has been adopted by several authors, including \cite{allais:LIPIcs.TYPES.2017.1,CervesatoHodasPfenning}.
In fact, all the systems similar to that of Hodas and Miller that we are aware of stem from Hodas' and Miller's original work.
The notion of a fragmentary environment is based on the idea of keeping a trace of resources that have been consumed.
This idea has been developed and exploited by McBride to provide adequate typing rules for linear dependent
products~\cite{McBride:IGPON}.
The only other typing system we know of that combines the same three ingredients as ours is the system developed by
Allais \cite{allais:LIPIcs.TYPES.2017.1}.
As a result, his system is quite similar to ours. However, he does not prove the same properties that we do.
In particular, he does not establish subject reduction in its full generality, but only for a given reduction strategy.

The work we have presented in this paper has been carried out as part of the development of ACGtk,
the Abstract Categorical Grammar support system~\cite{pogodalla:ACGtk}.
Abstract Categorical Grammars are type-theoretic grammars based on linear logic~\cite{deGroote:ACL01}.
They can be seen as the freely generated case of Mellies' and Zeilberger's type refinement systems~\cite{MelliesZeilberger}.
Within these grammars, the abstract parse structures are represented by linear lambda-terms, and parsing amounts
to a proof-search problem in linear logic.

Unlike the usual typing system, which uses an additive typing rule for the application, the typing system we have developed
enjoys an interesting backward chaining interpretation. Consider the problem of constructing (if any) a linear $\lambda$-term
$t$ of type $\alpha$ with respect to a typing environment $\Gamma$. A typing judgement of the form
\begin{align*}
\tytm{\Gamma}{t}{\alpha}{\Delta}
\end{align*}
can be interpreted as $\Gamma$ and $\alpha$ being the input of the proof-search problem, and $t$ and $\Delta$ being its
output ($t$ being the proof of $\alpha$ and $\Delta$ being the resources not consumed in constructing $t$).

Our typing system is therefore suitable for proof search, including the generation of proofs in the form of lambda-terms
in de Bruijn notation.
We intend, in future work, to take advantage of this feature and develop proof-search algorithms
for a fragment of exponential multiplicative linear logic.
This will require possible extensions of our typing system.


\bibliography{references.bib}

\begin{thebibliography}{10}
\providecommand{\bibitemdeclare}[2]{}
\providecommand{\surnamestart}{}
\providecommand{\surnameend}{}
\providecommand{\urlprefix}{Available at }
\providecommand{\url}[1]{\texttt{#1}}
\providecommand{\href}[2]{\texttt{#2}}
\providecommand{\urlalt}[2]{\href{#1}{#2}}
\providecommand{\doi}[1]{doi:\urlalt{https://doi.org/#1}{#1}}
\providecommand{\eprint}[1]{arXiv:\urlalt{https://arxiv.org/abs/#1}{#1}}
\providecommand{\bibinfo}[2]{#2}

\bibitemdeclare{inproceedings}{allais:LIPIcs.TYPES.2017.1}
\bibitem{allais:LIPIcs.TYPES.2017.1}
\bibinfo{author}{G.~\surnamestart Allais\surnameend} (\bibinfo{year}{2019}):
  \emph{\bibinfo{title}{{Typing with Leftovers - A mechanization of
  Intuitionistic Multiplicative-Additive Linear Logic}}}.
\newblock In \bibinfo{editor}{A.~\surnamestart Abel\surnameend},
  \bibinfo{editor}{F.~\surnamestart Nordvall~Forsberg\surnameend} \&
  \bibinfo{editor}{A.~\surnamestart Kaposi\surnameend}, editors: {\slshape
  \bibinfo{booktitle}{23rd International Conference on Types for Proofs and
  Programs (TYPES 2017)}}, {\slshape \bibinfo{series}{Leibniz International
  Proceedings in Informatics (LIPIcs)}} \bibinfo{volume}{104},
  \bibinfo{publisher}{Schloss Dagstuhl -- Leibniz-Zentrum f{\"u}r Informatik},
  \bibinfo{address}{Dagstuhl, Germany}, pp. \bibinfo{pages}{1:1--1:22},
  \doi{10.4230/LIPIcs.TYPES.2017.1}.

\bibitemdeclare{inproceedings}{Benton:TLCA93}
\bibitem{Benton:TLCA93}
\bibinfo{author}{N.~\surnamestart Benton\surnameend},
  \bibinfo{author}{G.~\surnamestart Bierman\surnameend},
  \bibinfo{author}{V.~\surnamestart De~Paiva\surnameend} \&
  \bibinfo{author}{M.~\surnamestart Hyland\surnameend} (\bibinfo{year}{1993}):
  \emph{\bibinfo{title}{A term calculus for intuitionistic linear logic}}.
\newblock In \bibinfo{editor}{M.~\surnamestart Bezem\surnameend} \&
  \bibinfo{editor}{J.F. \surnamestart Groote\surnameend}, editors: {\slshape
  \bibinfo{booktitle}{International Conference on Typed Lambda Calculi and
  Applications, TLCA'93}}, {\slshape \bibinfo{series}{Lecture Notes in Computer
  Science}} \bibinfo{volume}{664}, \bibinfo{publisher}{Springer}, pp.
  \bibinfo{pages}{75--90}, \doi{10.1007/BFb0037099}.

\bibitemdeclare{article}{deBruijn:indices}
\bibitem{deBruijn:indices}
\bibinfo{author}{N.G. \surnamestart de~Bruijn\surnameend}
  (\bibinfo{year}{1972}): \emph{\bibinfo{title}{Lambda calculus notations with
  nameless dummies, a tool for automatic formula manipulation, with an
  application to the {C}hurch-{R}osser theorem}}.
\newblock {\slshape \bibinfo{journal}{Indagationes Mathematicae}}
  \bibinfo{volume}{34}, pp. \bibinfo{pages}{381--392},
  \doi{10.1016/1385-7258(72)90034-0}.

\bibitemdeclare{article}{CervesatoHodasPfenning}
\bibitem{CervesatoHodasPfenning}
\bibinfo{author}{I.~\surnamestart Cervesato\surnameend}, \bibinfo{author}{J.~S.
  \surnamestart Hodas\surnameend} \& \bibinfo{author}{F.~\surnamestart
  Pfenning\surnameend} (\bibinfo{year}{2000}): \emph{\bibinfo{title}{Efficient
  resource management for linear logic proof search}}.
\newblock {\slshape \bibinfo{journal}{Theoretical Computer Science}}
  \bibinfo{volume}{232}(\bibinfo{number}{1}), pp. \bibinfo{pages}{133--163},
  \doi{10.1016/S0304-3975(99)00173-5}.

\bibitemdeclare{article}{Church:STT}
\bibitem{Church:STT}
\bibinfo{author}{A.~\surnamestart Church\surnameend} (\bibinfo{year}{1940}):
  \emph{\bibinfo{title}{A Formulation of the Simple Theory of Types}}.
\newblock {\slshape \bibinfo{journal}{Journal of Symbolic Logic}}
  \bibinfo{volume}{5}, pp. \bibinfo{pages}{56--68}, \doi{10.2307/2266170}.

\bibitemdeclare{article}{Girard:LL}
\bibitem{Girard:LL}
\bibinfo{author}{J.-Y. \surnamestart Girard\surnameend} (\bibinfo{year}{1987}):
  \emph{\bibinfo{title}{Linear Logic}}.
\newblock {\slshape \bibinfo{journal}{Theoretical Computer Science}}
  \bibinfo{volume}{50}, pp. \bibinfo{pages}{1--102},
  \doi{10.1016/0304-3975(87)90045-4}.

\bibitemdeclare{inproceedings}{deGroote:ACL01}
\bibitem{deGroote:ACL01}
\bibinfo{author}{Ph. \surnamestart de~Groote\surnameend}
  (\bibinfo{year}{2001}): \emph{\bibinfo{title}{{Towards Abstract Categorial
  Grammars}}}.
\newblock In: {\slshape \bibinfo{booktitle}{Association for Computational
  Linguistics, 39th Annual Meeting and 10th Conference of the European Chapter,
  Proceedings of the Conference}}, pp. \bibinfo{pages}{148--155},
  \doi{10.3115/1073012.1073045}.

\bibitemdeclare{inproceedings}{pogodalla:ACGtk}
\bibitem{pogodalla:ACGtk}
\bibinfo{author}{M.~\surnamestart Guillaume\surnameend},
  \bibinfo{author}{S.~\surnamestart Pogodalla\surnameend} \&
  \bibinfo{author}{V.~\surnamestart Tourneur\surnameend}
  (\bibinfo{year}{2024}): \emph{\bibinfo{title}{{ACGtk: A Toolkit for
  Developing and Running Abstract Categorial Grammars}}}.
\newblock In \bibinfo{editor}{J.~\surnamestart Gibbons\surnameend} \&
  \bibinfo{editor}{D.~\surnamestart Miller\surnameend}, editors: {\slshape
  \bibinfo{booktitle}{{Functional and Logic Programming. 17th International
  Symposium, FLOPS 2024}}}, {\slshape \bibinfo{series}{Lecture Notes in
  Computer Science}} \bibinfo{volume}{14659}, \bibinfo{publisher}{{Springer}},
  pp. \bibinfo{pages}{13--30}, \doi{10.1007/978-981-97-2300-3_2}.

\bibitemdeclare{article}{Hodas-Miller}
\bibitem{Hodas-Miller}
\bibinfo{author}{J.S. \surnamestart Hodas\surnameend} \&
  \bibinfo{author}{D.~\surnamestart Miller\surnameend} (\bibinfo{year}{1994}):
  \emph{\bibinfo{title}{{Logic Programming in a Fragment of Intuitionistic
  Linear Logic}}}.
\newblock {\slshape \bibinfo{journal}{Information and Computation}}
  \bibinfo{volume}{110}(\bibinfo{number}{2}), pp. \bibinfo{pages}{327--365},
  \doi{10.1006/inco.1994.1036}.

\bibitemdeclare{inproceedings}{McBride:IGPON}
\bibitem{McBride:IGPON}
\bibinfo{author}{C.~\surnamestart McBride\surnameend} (\bibinfo{year}{2016}):
  \emph{\bibinfo{title}{I Got Plenty o' Nuttin'}}.
\newblock In \bibinfo{editor}{S.~\surnamestart Lindley\surnameend},
  \bibinfo{editor}{C~\surnamestart McBride\surnameend}, \bibinfo{editor}{Ph.~W.
  \surnamestart Trinder\surnameend} \& \bibinfo{editor}{D.~\surnamestart
  Sannella\surnameend}, editors: {\slshape \bibinfo{booktitle}{A List of
  Successes That Can Change the World - Essays Dedicated to Philip Wadler on
  the Occasion of His 60th Birthday}}, {\slshape \bibinfo{series}{Lecture Notes
  in Computer Science}} \bibinfo{volume}{9600}, \bibinfo{publisher}{Springer},
  pp. \bibinfo{pages}{207--233}, \doi{10.1007/978-3-319-30936-1_12}.

\bibitemdeclare{inproceedings}{MelliesZeilberger}
\bibitem{MelliesZeilberger}
\bibinfo{author}{P.{-}A. \surnamestart Melli{\`{e}}s\surnameend} \&
  \bibinfo{author}{N.~\surnamestart Zeilberger\surnameend}
  (\bibinfo{year}{2015}): \emph{\bibinfo{title}{Functors are Type Refinement
  Systems}}.
\newblock In: {\slshape \bibinfo{booktitle}{Proceedings of the 42nd Annual
  {ACM} {SIGPLAN-SIGACT} Symposium on Principles of Programming Languages,
  {POPL} 2015}}, pp. \bibinfo{pages}{3--16}, \doi{10.1145/2676726.2676970}.

\end{thebibliography}

\end{document}